\documentclass[journal]{IEEEtran}
\normalsize
\usepackage{epstopdf}
\usepackage{cite}
  \usepackage[dvips]{graphicx}
  \graphicspath{{./figures/}}
\usepackage[cmex10]{amsmath}
\usepackage{amssymb}
\usepackage{setspace}
\usepackage{algorithmic}

\usepackage{array}

\usepackage{caption}
\usepackage{subcaption}
\usepackage{fixltx2e}
\usepackage{stfloats}
\usepackage{boxedminipage}

\usepackage{times,color}

\renewcommand{\P}{\mathbb{P}}
\newcommand{\mc}{\mathcal}
\newcommand{\mb}{\mathbb}
\newcommand{\ms}{\mathsf}
\newcommand{\E}{\mathbb{E}}

\newtheorem{theorem}{Theorem}

\newtheorem{definition}{Definition}

\newtheorem{corollary}{Corollary}

\begin{document}

\title{Variable-Length Coding with Feedback: Finite-Length Codewords and Periodic Decoding}
\author{
\IEEEauthorblockN{Tsung-Yi Chen, Adam R. Williamson and Richard D. Wesel}\\
\IEEEauthorblockA{Department of Electrical Engineering\\
University of California, Los Angeles\\
Email: tychen@ee.ucla.edu; adamroyce@ucla.edu; wesel@ee.ucla.edu}%
\thanks{This research was supported by National Science Foundation Grant CIF CCF 1162501.}
}
\maketitle
\begin{abstract}
Theoretical analysis has long indicated that feedback improves the error exponent but not the capacity of single-user memoryless channels. Recently Polyanskiy et al. studied the benefit of variable-length feedback with termination (VLFT) codes in the non-asymptotic regime. In that work, achievability is based on an infinite length random code and decoding is attempted at every symbol. The coding rate backoff from capacity due to channel dispersion is greatly reduced with feedback, allowing capacity to be approached with surprisingly small expected latency. This paper is mainly concerned with VLFT codes based on {\em finite-length} codes and decoding attempts only at {\em certain specified decoding times}. The penalties of using a finite block-length $N$ and a sequence of periodic decoding times are studied. This paper shows that properly scaling $N$ with the expected latency can achieve the same performance up to constant terms as with $N = \infty$. The penalty introduced by periodic decoding times is a linear term of the interval between decoding times and hence the performance approaches capacity as the expected latency grows if the interval between decoding times grows sub-linearly with the expected latency.
\end{abstract}


\section{Introduction}
While feedback cannot increase the capacity of a memoryless channel, it can significantly reduce the complexity of encoding and decoding at rates below capacity. 
The error exponent results of \cite{Schalkwijk_1966_1, Schalkwijk_1966_2, Kramer_1969, Zig_1970, Nakiboglu_2008, Gallager_2010} suggest that feedback can be used to reduce the average block-length (or expected latency) required to approach capacity.  As a practical demonstration, \cite{Chen_2010_ITA} showed that using an incremental redundancy (IR) scheme with feedback allows short convolutional codes to deliver bit error rate performance comparable to a long-block-length turbo code, but with lower latency.  The demonstration of \cite{Chen_2010_ITA} qualitatively agrees with the error exponent analysis in \cite{Schalkwijk_1966_1, Schalkwijk_1966_2, Kramer_1969, Zig_1970, Nakiboglu_2008, Gallager_2010}. 

Because of its asymptotic perspective, the error exponent theory does not provide an accurate prediction of short-block-length performance. For example, Yamamoto and Itoh \cite{Yamamoto_1979} showed that the optimal Burnashev error exponent \cite{Burnashev_1976} is achievable by a two-phase ARQ coding scheme.  However, at short block-lengths (i.e. for a small average number of channel uses) a considerable performance gap exists between ARQ and a well-designed IR scheme. Polyanskiy et al. \cite{PolyIT11} analyzed the benefit of feedback in the non-asymptotic regime and provide quantitative characterizations for short expected latency. They show that capacity can be closely approached in hundreds of symbols rather than thousands using variable-length feedback codes with termination (VLFT codes), a form of IR. 

The analysis of VLFT in \cite{PolyIT11} assumes an underlying codebook with infinite-length codewords and decoding is attempted at every symbol so that the communication may be concluded after any given channel use. In practice, it may only be possible to use a codebook with finite-length codewords.  It may also be possible only to attempt decoding (and thus conclude communication) after channel uses that come at the end of a group of symbols because of packetization, decoding delays, and round-trip propagation times.  With these practical issues in mind, this paper studies the penalties that occur when the codebook is limited to finite-length codewords and/or decoding (and therefore termination) is only possible at periodic intervals rather than at every symbol.

\section{Previous Work and Main Results}
\label{sec:PreWorkMainContrib} 
\subsection{Previous Work}
\label{sec:PreWork}
We will consider discrete memoryless channels (DMC) throughout the paper and use the following notation: $x^n = (x_1, x_2, \dots, x_n)$ denotes an $n$-dimensional vector, $x_j$ the $j$th element of $x^n$, and $x_i^j$ the $i$th to $j$th elements of $x^n$. We denote random variables by capitalized letters unless otherwise stated. The input and output alphabets are $\mc{X}$ and $\mc{Y}$ respectively.  Let the input and output product spaces be $\ms{X} = \mc{X}^n, \ms{Y} = \mc{Y}^n$ respectively. A channel is characterized by a conditional distribution $P_{\ms{Y}|\ms{X}} = \prod_{i = 1}^n P_{Y_i|X_i}$ where the equality holds because the channel is memoryless. For codes that make use of a noiseless feedback link, we consider causal channels $\{P_{Y_i|X_1^iY^{i-1}}\}_{i = 1}^{\infty}$ and additionally focus on causal memoryless channels $\{P_{Y_i|X_i}\}_{i = 1}^{\infty}$. 

We are interested in zero-error communication with feedback in this paper and will therefore focus on the paradigm of VLFT coding. In order to be self-contained, we state the definition of VLFT codes in \cite{PolyIT11}:
\begin{definition}
An $(\ell,M,\epsilon)$ variable-length feedback code with termination (VLFT code) is defined as:
	\begin{enumerate} 
	\item A common random variable (r.v.) $U \in \mc{U}$ with a probability distribution $P_U$ revealed to both transmitter and receiver before the start of transmission.
	\item A sequence of encoders $f_n:\mc{U}\times\mc{W}\times \mc{Y}^{n-1} \rightarrow \mc{X}$ that defines the channel inputs $X_n = f_n(U,W,Y^{n-1})$.  Here $W$ is the message r.v. uniform in $\{1, \dots, M\}$.
	\item A sequence of decoders $g_n: \mc{U}\times\mc{Y}^n \rightarrow \mc{W}$ providing the estimate of $W$ at time $n$.
	\item A stopping time $\tau \in \mb{N}$ w.r.t. the filtration $\mc{F}_n = \sigma\{U, Y^n, W\}$ such that: 
	\begin{align}
	\E[\tau]\leq \ell.
	\end{align}
	\item The final decision $\hat{W} = g_\tau(U, Y^\tau)$ must satisfy:
	\begin{align}
	\P[\hat{W} \ne W]\leq \epsilon.
	\end{align}
	\end{enumerate}
\end{definition}

As observed in \cite{PolyIT11}, the setup of VLFT is equivalent to augmenting each channel with a special use-once input symbol, the termination symbol, that has infinite reliability. This assumption captures the fact that many practical systems communicate control signals in the upper protocol layers and the termination symbol effectively separates the control issue from the physical channel. The benefit of the infinitely reliable control signal can cause the VLFT achievable rate to be larger than that of the original feedback channel capacity because what would have been a decoding error without feedback becomes a codeword ``erasure'' under VLFT.  

The class of fixed-to-variable codes \cite{Verdu_10}, or FV codes, is a special class of VLFT codes that satisfies the following conditions:
\begin{align}
&f_n(U,W,Y^{n-1}) = f_n(U,W)
\\
&\tau = \inf\{ n \geq 1: g_n(U,Y^n) = W\}.
\end{align}
Such codes are zero-error VLFT codes and only use feedback to stop the transmission. Fountain codes and families of rate-compatible codes used with an IR scheme are examples of such codes. This class of codes is widely used in practical systems and will be the main focus of this paper. 

Let the finite dimensional distribution of $(X^n, \bar{X}^n, Y^n)$ be: 
\begin{align}
&P_{X^nY^n\bar{X}^n}(x^n,y^n,\bar{x}^n) 
\\
&= P_{X^n}(x^n)P_{X^n}(\bar{x}^n)\prod_{j = 1}^nP_{Y_j|X^jY^{j-1}}(y_j|x^j, y^{j-1})\, ,
\end{align}
i.e. the distribution of $\bar{X}^n$ is identical to $X^n$ but independent of $Y^n$. The information density $i(x^n;y^n)$ is defined as
\begin{align}
i(x^n; y^n) &= \log \frac{dP_{X^nY^n}(x^n,y^n)}{d(P_{X^n}(x^n)\times P_{Y^n}(y^n))} 
\\
&= \log \frac{dP_{Y^n|X^n}(y^n|x^n)}{dP_{Y^n}(y^n)}.
\end{align}
The following is the achievability result in \cite{PolyIT11}:
\begin{theorem}[\cite{PolyIT11}, Thm. 10]
\label{thm:RC_Achieve}
Fixing $M > 0$, there exists an $(\ell, M, 0)$ VLFT code with
\begin{align}
\label{eqn:AchevVLFT}
\ell &\leq \sum_{n = 0}^{\infty} \xi_n
\end{align}
where $\xi_n$ is the following expectation: 
\begin{equation} \label{eqn:Xi_n}
\E\left[\min\left\{1, (M-1)\Pr[i(X^n;Y^n)\leq i(\bar{X}^n;Y^n)|X^nY^n]\right\}\right].
\end{equation}
The expression above is referred to as random coding union (RCU) bound.
We take $i(X^0;Y^0) = 0$ and hence $\xi_0 = 1$.  Additionally, from the proof of \cite[Thm. 11]{PolyIT11}, we have: 
\begin {equation}
\label{eqn:VLFT_DT}
\xi_n \leq \E\left[ \exp\left\{-[i(X^n; Y^n) - \log \gamma]^+\right \}\right]. 
\end{equation}
\end{theorem}
\subsection{Problem Statement and Main Results}
\label{sec:MainConstribution}
Following the VLFT framework of \cite{PolyIT11}, this paper studies the following problems:
\begin{enumerate}
\item[(i)] Finite-length codeword penalty for VLFT (FV) codes: The random coding approach in \cite{PolyIT11} generates random codebooks in an infinite product space (i.e. with codewords of infinite length). We study the performance penalty incurred by using random codebooks with a finite block-length.
\item[(ii)] The penalty associated with limitations on decoding times:  We study the performance penalty incurred when decoding is only allowed after every $I$ symbols are received, i.e. periodic decoding times. The case where the decoding times can be an arbitrary set of increments $\{I_j\}_{j = 1}^{m}$ is studied in \cite{ChenDraft2012}.
\end{enumerate}

For the rest of the paper we only consider channels with essentially bounded information density $i(X;Y)$. Define the fundamental transmission limit of a VLFT code with finite block-length $N$ and uniform increment $I$ as follows:
\begin{definition}
Let $M^*_t(\ell, N, I, \epsilon)$ be the maximum integer $M$ such that there exist an $(\ell, M, N, I, \epsilon)$ VLFT code based on a code with block-length $N$ and a decoder that only attempts decoding every $I$ symbols. For zero-error codes where $\epsilon = 0$ we denote the maximum $M$ as $M^*_t(\ell, N, I)$ and for zero-error codes with $I=1$ (i.e. decoding attempts after every received symbol) we denote the maximum $M$ as $M^*_t(\ell, N)$.
\end{definition}

All of the results that follow assume an arbitrary but fixed channel $\{P_{Y_j|X_j}\}_{j = 1}^N$ and a process $\{X_j\}_{j = 1}^{N}$ taking values in $\mc{X}$ where $N$ could be set as infinity. Our main asymptotic result is the following expansion for a stationary DMC:
\begin{theorem}
\label{thm:MainAsympResult2}
Choosing $N = \ell + \Omega(\log\ell)$ for a stationary DMC with capacity $C$, we have:
\begin{align}
\label{eq:Thm2}
\log M^*_t(\ell, N, I) \geq \ell C - O(I)\, .
\end{align}
\end{theorem}
Specifically, if we choose $N > \ell + \frac{\log(\ell+1) + \log e}{C}$ and have decoding attempts separated by an increment $I$, then the expansion is the same as the case with $N = \infty$ and the constant term depends on the choice of the increment $I$. The proof is provided in Section~\ref{sec:VLFT_TimeLimit2}.

Of course, for practical applications that apply feedback to obtain reduced latency, the non-asymptotic behavior is critical.  Numerical results on a binary symmetric channel demonstrate that properly selected values of $N$ and $I$ can yield excellent expected throughput with expected latency on the order of $200$ symbols.

The rest of the paper is organized as follows:  Section~\ref{sec:VLFT_FiniteLength} investigates the penalty incurred by using VLFT codes based on finite block-lengths. Section~\ref{sec:VLFT_TimeLimit1} studies the penalty incurred by limiting decoding attempts, and Section~\ref{sec:VLFT_TimeLimit2} studies the penalty when both limitations are applied. Section~\ref{sec:NumResults} gives numerical results for a binary symmetric channel. Section~\ref{sec:Conclusion} concludes the paper.


\section{Finite Block-Lengths and Limited Decoding}
In \cite{PolyIT11}, attention was focused on $(\ell, M, N, I, \epsilon)$ VLFT codes with $N=\infty$ and $I=1$. This section studies the penalties associated with using finite $N$ and $I\geq 1$.   We focus on the $\epsilon = 0$ case.  The random coding framework of \cite{PolyIT11} is retained.  We focus on achievability results under these constrained scenarios using proofs based on random FV codes.  The general converse established in \cite{PolyIT11} still applies since these additional constraints can only further limit performance.

\subsection{The Finite-Block-Length Limitation}
\label{sec:VLFT_FiniteLength} 
This subsection investigates $(\ell, M, N, I, \epsilon)$ VLFT codes with finite $N$ but retains decoding at every symbol ($I=1$). 
FV codes (as described in Section \ref{sec:PreWork}) are employed so that encoding does not depend on the feedback except that feedback indicates when it is the time to terminate transmission. 

Letting $\zeta_j$ be the marginal error event at the $j$th transmission, the expected latency $\E[\tau]$ is given as:
\begin{align}
\label{eqn:SumsOfPtau}
\E[\tau]  &= \sum_{n = 1}^{\infty} n \P[\tau = n]
\\
&= \sum_{n \geq 1} \P[\tau > n]
\\
\label{eqn:SumsOfJoints}
&=\sum_{n \geq 1} \P\left[\bigcap_{j = 1}^{n}\zeta_j\right] \, .
\end{align}

Consider a code $\mc{C}_N$ with finite block-length $N$ where each element is a length-$N$ $\mc{X}$-valued string.  Achievability results for an $(\ell, M, N, 1, \epsilon)$ ``truncated'' VLFT code follow from a random coding argument. In particular we have the following: 
\begin{theorem}
\label{thm:FiniteVLFT}
For any $M > 0$ there exists an $(\ell, M, N, 1, \epsilon)$ truncated VLFT code with
\begin{align}
\ell &\leq \sum_{n = 0}^{N-1} \xi_n
\\
\epsilon &\leq \xi_N.
\end{align}
where $\xi_n$ is the same as \eqref{eqn:Xi_n}. 
The proof is in the appendix. \end{theorem}

Achievability results for $\epsilon=0$ can be obtained using an $(\ell, M, N, 1, 0)$ ``repeated'' VLFT code, which modifies the encoder and decoder pairs with an ARQ-type repetition.  When the block-length-$N$ codeword is exhausted without successful decoding,  the transmission process starts from scratch discarding the previous received symbols.  Using the original $N$ symbols through, for example, Chase code combining would be beneficial, but this is not necessary for our achievability result.  Specifically, we have the following result for a zero-error repeated VLFT code with a finite block-length $N$:
\begin{theorem}
\label{thm:FiniteFV}
For every $M > 0$ there exists an $(\ell, M, N, 1, 0)$ repeated VLFT code such that
\begin{align}
\label{eqn:FiniteFV}
\ell &\leq \frac{1}{(1-\xi_N)} \sum_{n = 0}^{N-1}\xi_n
\end{align}
where $\xi_n$ is the same as \eqref{eqn:Xi_n}.  The proof is in the appendix.
\end{theorem}

Note that this is an FV code based on a finite-length codebook rather than an infinite one. The penalty of using a codebook with finite length is made clear in the following theorem and its corollary:
\begin{theorem}
\label{thm:VLFTExpandFinite}
For an $(\ell, M, N, 1, 0)$ repeated VLFT code with $N = \Omega(\log M)$, we have the following expansion for a stationary DMC  with capacity $C$:
\begin{align}
\label{eqn:ellExpansion1}
\ell \leq \frac{\log M}{C} + O(1)\, .
\end{align}
Let  $C_\Delta = C- \Delta, \Delta > 0$ and $N = \log M / C_\Delta$. The correction term is upper bounded as follows:
\begin{align}
\label{eqn:ellExpansion2}
O(1) \leq \frac{b_2 \log M}{C(M^{b_3/C_\Delta} - b_2)} 
+ \frac{b_0\log M}{C_\Delta M^{b_1\Delta/C_\Delta}} + a
\end{align}
where $a$ depends on the mean and uniform bound of $i(X;Y)$, and $b_j$'s are constants related to $\Delta$ and $M$. 
\end{theorem}
This choice of $N$ has residual terms decaying with $M$ very slowly. However, our numerical results indicate that this decay is fast enough for excellent performance in the short-block-length regime. 

We define a pair of random walks to simplify the proofs:
\begin{align}
\label{eqn:RW1}
S_n &\triangleq i(X^n;Y^n)
\\
\label{eqn:RW2}
\bar{S}_n &\triangleq i(\bar{X}^n;Y^n)\,.
\end{align}
For any measurable function $f$ we have the property:
\begin{align}
\label{eqn:tilting}
\E[f(\bar{X}^n;Y^n)] = \E[f(X^n;Y^n)\exp\{-S_n\}].
\end{align}
Observe that $S_n$ and $\bar{S}^n$ are sums of i.i.d. r.v.s with positive and negative means $\E[i(X;Y)] = C$ and $\E[i(\bar{X};Y)] = -L$ ($L$ is the lautum information \cite{Palomar08}) respectively. In particular $\{S_n - nC\}_n$ is a bounded martingale and hence by Doob's optional stopping theorem we have for a stopping time $\tau$:
\begin{align}
\label{eqn:DoobOpt}
\E[S_\tau] = C\E[\tau]\,.
\end{align}
These properties are utilized in the following proofs. 
\begin{proof}
Following the definition of \eqref{eqn:RW1} and \eqref{eqn:RW2}, we first weaken the RCU bound by \eqref{eqn:VLFT_DT} and choosing $\gamma = M$:
\begin{align}
&\E\left[\min\left\{(1,(M-1)\P[\bar{S}_n\geq S_n|X^nY^n]\right\}\right] 
\\
&\leq \E\left[\exp\{-[S_n - \log M]^+\}\right]\,.
\end{align} 
Then from Thm.\ref{thm:FiniteFV} we have:
\begin{align}
\ell \leq \frac{1}{(1-\P[\zeta_N])}\sum_{n = 0}^{N-1}\E\left[\exp\left\{-[S_n - \log M]^+\right\}\right]\, .
\end{align}
Consider an auxiliary stopping time $\tau'$ w.r.t. the filtration $\mc{F}_n = \sigma\{X^n, \bar{X}^n, Y^n\}$:
\begin{align}
\tau' = \inf\{n \geq 0: S_n \geq \log M\}\wedge N\,.
\end{align}
Denoting $\E[X; E] = \E[X 1_E]$ where $1_E$ is the indicator function of the set $E$, we have:
\begin{align}
	&\sum_{n = 0}^{N-1}\E\left[\exp\{-[S_n - \log M]^+\}\right] 
\\
\nonumber
	&= \E\left[\tau' + \sum_{k = 0}^{N-1-\tau'}\exp\left\{-[S_{k+\tau'} - \log M ]^+\right\};\tau' < N\right]
\\
	&\quad + N\P[\tau' \geq N]\,.
\end{align}
On $\{\tau' < N\}$ we have $i(X^{\tau'};Y^{\tau'}) \geq \log M$ and hence:
\begin{align}
&\left[S_{k+\tau'} - \log M\right]^+ 
= \left[S_{k+\tau'} - S_{\tau'} + S_{\tau'} - \log M\right]^+
\\
&\geq  \left[i(X^{k+\tau'};Y^{k+\tau'}) - i(X^{\tau'};Y^{\tau'}) \right]^+
\\
&= \left[i(X^{k};Y^{k})\right]^+ 
\end{align}
where the last equality is true almost surely by the strong Markov property of random walks. It then follows that:
\begin{align}
\nonumber
&\sum_{n = 0}^{N-1}\E\left[\exp\{-[S_n - \log M]^+\}\right] 
\\
\nonumber
&\leq 
\E\left[\tau' + \sum_{k = 0}^{N-1-\tau'}\exp\{-\left[S_k\right]^+ \}; \tau' < N\right] 
\\
&\quad+ N\P[ \tau' \geq N ]\,.
\end{align}
Observe that $S_n$ and $\bar{S}^n$ are sums of i.i.d. r.v.s with positive and negative means respectively. Thus by Chernoff inequality we have that:
\begin{align}
\E\left[e^{-[i(X^k;Y^k)]^+} \right]
&= \P\left[S_k > 0\right] + \P\left[\bar{S}_k\leq 0\right]
\\
&\leq a_1e^{-ka_2} \, ,
\end{align}
wher first equality follows from \eqref{eqn:tilting}.
Thus there is a constant $a_3 > 0$ such that:
\begin{align}
\sum_{k = 1}^{N-1-\tau'}\E\left[\exp\{-[S_k]^+\} \right]
&\leq \sum_{k = 1}^{N-1}\E\left[\exp\{-[S_k]^+\} \right] 
\\\label{eqn:a3}
&\leq a_3\, .
\end{align}
We assume that $i(X^n; Y^n)$ has bounded jumps, and hence on the set $\{\tau' < N\}$ there is a constant $a_4$ such that
\begin{align}
i(X^{\tau'};Y^{\tau'}) - \log M \leq a_4 C \,. 
\end{align}
Therefore from \eqref{eqn:DoobOpt} we have:
\begin{align}
\label{eqn:a4}
\E[\tau'] \leq \frac{\log M }{C} + a_4 \,.
\end{align}
Letting $a_5=a_3+a_4$ and combining \eqref{eqn:a3} and \eqref{eqn:a4} we have:
\begin{align}
\label{eqn:EllBoudnAtN}
\ell \leq (1-\P[\zeta_N])^{-1}\left(\frac{\log M}{C} + N\P[\tau' \geq N] + a_5 \right) \, .
\end{align} 
Let $C_\Delta = C - \Delta$. For a fixed $M$, we can take $N = \log M / C_\Delta$ for a constant $\Delta > 0$ such that: 
\begin{align}
\label{eqn:ErrorAtN}
\P[\zeta_N]\leq b_2\exp(-N b_3)\,. 
\end{align}
Again by noting that $S_n$ is a sum of i.i.d. r.v.'s with mean $C$, for $NC > \log M$ we have by Chernoff inequality that: 
\begin{align}
\P[\tau' \geq N] & = \P[S_N < \log M]
\\
&\leq b_0 \exp\left\{-b_1 N\left(C - \frac{\log M }{N}\right)\right\} 
\\ \label{eqn:ChernoffAtN}
& = b_0 \exp\left\{-b_1\Delta\frac{\log M}{C_\Delta}\right\}.
\end{align}
Combining \eqref{eqn:EllBoudnAtN}, \eqref{eqn:ErrorAtN} and \eqref{eqn:ChernoffAtN}  we have the following for $\ell$:
\begin{align}
\label{eqn:LogScaling1}
\ell \leq \left( \frac{\log M}{C} 
+ \frac{ b_0\log M}{C_\Delta M^{b_1\Delta/C_\Delta}}+a_5\right) (1-\P[\zeta_N])^{-1}\,.
\end{align}
Notice that we are only interested in the first two terms of the expansion $(1-x)^{-1} = 1 + x + x^2 + \dots$ on $[0,1)$. Thus
\begin{align}
\ell &\leq \frac{\log M}{C} + \frac{ b_0\log M}{C_\Delta M^{b_1\Delta/C_\Delta}}+ \frac{ b_2\log M}{C(M^{b_3/C_\Delta}-b_2)} + a_6
\end{align}
for some $a_6> 0$. Hence for $M$ large enough we have \eqref{eqn:ellExpansion1}.
\end{proof}


An expansion of $\log M_t^*(\ell, N)$ requires $N$ growing with $\ell$. The components of the correction term in Thm. \ref{thm:VLFTExpandFinite}, however, depend on both $N$ (as  $\log M / C_\Delta$) and $M$. Indeed for a fixed $\ell$, all $M$ satisfying \eqref{eqn:ellExpansion1} and \eqref{eqn:ellExpansion2} are achievable.  The argument we make below is that for any fixed constant $c_0 > 0$, there is an $\ell_0$ that depends logarithmically on $c_0^{-1}$ such that the expansion $\log M_t^* \geq C \ell - c_0$ is true for all $\ell \geq \ell_0$. 
We first invoke the converse for an $(\ell, M, \infty, 1, 0)$ VFLT code:
\begin{theorem}[\cite{PolyIT11}, Thm. 11] \label{thm:poly11}
Given a stationary DMC with capacity $C$ we have the following for an $(\ell, M,\infty, 1, 0)$ VLFT code:
\begin{align}
\log M_t^* \leq \ell C + \log(\ell+1) + \log e\,.
\end{align}
\end{theorem}
Combining Thms. \ref{thm:VLFTExpandFinite} and \ref{thm:poly11} we have the following:
\begin{corollary}
\label{cor:ScalingForEll}
For an $(\ell, M, N, 1, 0)$ repeated VLFT code with $N = (1+\delta) \ell$ and a proper choice of $\delta > 0$, we have the following for a stationary DMC with capacity $C$:
\footnote{As opposed to the expression in \cite{PolyIT11}, we use a minus sign for $O(1)$ term to make the penalty clear.}
\begin{align}
\log M_t^*(\ell, N, I) \geq \ell C - O(1)\,.
\end{align}
\end{corollary}
\begin{proof}
We first choose $N$ to scale with $\ell$ with a factor $\delta$:
\begin{align}
N = (1+\delta)\ell\,.
\end{align}
Then by the converse we have:
\begin{align}
\frac{\log M}{N} &\leq C + \frac{\log(\ell+1)+\log e - \delta\ell C}{(1+\delta)\ell}
\\
&\leq C-\delta' \, .
\end{align}
The term $\delta'$ on the right is positive by setting:
\begin{align}
\delta > \frac{\log(\ell+1)+\log e}{\ell C}\,.
\end{align}
Again by Chernoff inequality we have:
\begin{align}
\P[\tau' \geq N] &= \P[S_N  < \log M  ]
\\
&\leq \P[S_N - NC < -N\delta' ]
\\
&\leq b_0'\exp\left\{-\ell(1+ \delta)b_1'\delta' \right\}\, .
\end{align}
Since $\delta$ is chosen such that $\log M / N$ is less than capacity, we also have \eqref{eqn:ErrorAtN}. By reordering \eqref{eqn:EllBoudnAtN} we have for some $b_2', b_3' > 0$ such that:
\begin{align}
\frac{\log M}{C} 
\geq \ell\left[1-b_2e^{-\ell(1+\delta)b_3} - (1+\delta)b_0'e^{-b_2'\ell}\right] - b_3' \,,
\end{align}
which implies $\log M_t^* \geq \ell C - O(1)$ for large enough $\ell$.
\end{proof}

To conclude the discussion of the penalty associated with finite block-length, we comment that $N$ only needs to be scaled properly, i.e. $(1+\delta)\ell$ for $\delta$ decreasing with $\ell$, to obtain the infinite-block-length expansion of $M^*_{t}(\ell, \infty)$ provided in \cite{PolyIT11}.  Thus, the restriction to a finite block-length $N$ does not restrict the asymptotic performance if $N$ is selected properly with respect to $\ell$.  However, the constant penalty is indeed different for infinite and finite $N$, which might not be negligible in the short-block-length regime. Still, our numerical results in Section \ref{sec:NumResults} indicate that relatively small values of $N$ can yield good results for short block-lengths.

\subsection{Limited, Regularly-Spaced, Decoding Attempts}
\label{sec:VLFT_TimeLimit1} 
This subsection investigates $(\ell, M, N, I, \epsilon)$ VLFT codes with $N=\infty$ but decoding attempted only at specified, regularly-spaced, symbols ($I>1$).  The first decoding time occurs after $n_1$ symbols (which could be larger than $I$) so that the decoding attempts are made at the times $n_j = n_1 + (j - 1)I$.  The relevant information density process $i(X^{n_j}; Y^{n_j})$ is on the subsequence $n_j = n_1 + (j - 1)I$.  The main result here is that the constant penalty now scales linearly with $I$:
\begin{theorem}
For an $(\ell, M, N, I, 0)$ VLFT code with uniform increments $I$ and $N = \infty$ we have the following expansion for a stationary DMC with capacity $C$:
\begin{align}
\log M^*_t(\ell, \infty, I) \geq \ell C - O(I)\,.
\end{align}
\end{theorem}
\begin{proof}
Consider the same random coding scheme as in Thm. \ref{thm:FiniteFV}, but now the auxiliary stopping time is given as $n_{\tau_0} = n_1 + (\tau_0-1)I$ where $\tau_0$ is also a stopping time given as:
\begin{align}
 \tau_0 = \inf\{j > 0: S_{n_j} = i(X^{n_j}; Y^{n_j}) \geq \log M\}\,.
\end{align}
The rest is similar to the proof of Thm.~\ref{thm:FiniteFV}:
\begin{align}
\ell &\leq n_1 + I\sum_{j = 1}^{\infty}\P[\zeta_{n_j}]
\\
&\leq n_1 + I\sum_{j = 1}^{\infty}\E\left[\exp\left\{-\left[S_{n_j} - \log M\right]^+\right\}\right]
\\
\nonumber
&\leq n_1 + I\E[\tau_0 - 1]
\\ 
&\quad+ \sum_{k = 0}^{\infty}\E\left[\exp\left\{-\left[S_{n_{\tau_0 + k}} - \log M\right]^+\right\}\right]
\\
&\leq \E[n_{\tau_0}] 
+ I\sum_{k = 0}^{\infty}\E\left[\exp\left\{-[i(X^{n_k};Y^{n_k})]^+\right\}\right]
\\
\label{eqn:UnifIncBd1}
&\leq \E[n_{\tau_0}] + I a_3
\\
\label{eqn:UnifIncBd2}
&\leq \frac{\log M}{C} + I a_4 + I a_3\,,
\end{align}
where \eqref{eqn:UnifIncBd1} follows by applying Chernoff inequality and \eqref{eqn:UnifIncBd2} is because the jumps of $i(X^{n_j};Y^{n_j})$ is bounded by $I\times a_4 C$ for some $a_4 > 0$. 
Reordering the equations gives the result. 
\end{proof}

In view of the theorem, the increment $I$ can grow slowly, e.g. $I = O(\log \ell)$ and can still permit an expected rate that approaches $C$ without the dispersion penalty. In the non-asymptotic regime, however, the penalty might not be negligible.   Our numerical results in Section~\ref{sec:NumResults} indicate that $I = \lceil\log_2 \log_2 M\rceil$ yields good results for short block-lengths.

\subsection{Finite Block-Length and Limited Decoding Attempts}
\label{sec:VLFT_TimeLimit2} 
This subsection investigates $(\ell, M, N, I, 0)$ (repeated) VLFT codes with {\em both} finite $N$ and $I>1$.   When these two limitations  are combined, a key parameter is $m$, the number of decoding attempts before the transmission process must start from scratch if successful decoding has not yet been achieved. The main result follows from combining the results of Sections \ref{sec:VLFT_FiniteLength} and \ref{sec:VLFT_TimeLimit1}.  Once $n_1$, $N$ and $I$ are specified, the value of $m$ is implied.  Specifically, we have the following theorem: 
\begin{theorem}
\label{thm:MainAsympResult1}
For an $(\ell, M, N, I, 0)$ VLFT code with $N = \Omega(\log M)$, we have the following for a stationary DMC with capacity $C$:
\begin{align}
\ell&\leq (1+\P[\zeta_N])^{-1}\frac{\log M}{C} +  \P[\tau_0 \geq m] + O(I)
\\
&\leq \frac{\log M}{C} + O(I) \,,
\end{align}
where $\tau_0$ is the stopping time in terms of the number of decoding attempts up to and including the first success. 
\end{theorem}
The proof is similar to Thm. \ref{thm:VLFTExpandFinite} and can be found in the appendix. The proof of Thm. \ref{thm:MainAsympResult2} now follows:
\begin{proof}[Proof of Thm. \ref{thm:MainAsympResult2}]
For an $(\ell, M, N, I, 0)$ VLFT code pick $N$ as follows:
\begin{align}
N = (1+\delta)\ell, \text{ where } \delta > \frac{\log(\ell+1)+\log e}{\ell C}\,.
\end{align}
The result follows by a similar argument as for Cor.~\ref{cor:ScalingForEll}. The restriction on the initial block-length $n_1$ only makes a constant difference. 
\end{proof}

\section{Numerical Results}
\label{sec:NumResults}
We give a numerical example of our results for a binary symmetric channel (BSC). 
For a BSC with transition probability $p$ we used the RCU bound in \cite{PolyIT10,PolyIT11}\footnote{We replace $(M-1)$ by $M$ for simplicity.}, which gives the following expression:
\begin{align*}
\xi_n \leq \sum_{t = 0}^{n} {n\choose t}p^t (1-p)^{n-t}\min\left\{1, M\sum_{j=0}^{t}{n\choose j} 2^{-n}\right\} \, .
\end{align*}

Fig.~\ref{fig:VLFT_FiniteBlock} shows the performance of VLFT codes over a BSC with $p = 0.0789$ with $N=\infty$ $N = \ell + \Omega(\log\ell)$, and $N = \log M/C_\Delta$. Since $\ell$ scales linearly with $\frac{\log M}{C}$, for the case of $\ell + \Omega(\log\ell)$ we choose $N$ to scale as:
\begin{align}
N = \frac{\log M}{C} + a \log\left(\frac{\log M}{C}\right) + b\,,
\end{align}
where $a, b > 0$ are constants to be chosen numerically. The numerical examples presented here use $a = 10$, $b = 30$. We choose $\Delta = 0.3C$ and $0.4C$, which are about $43\%$ and $67\%$ longer, respectively, than the block-length that corresponds to capacity. In other words, $N = 1.43\log M/C$ and $N = 1.67\log M/C$ respectively. As expected latency increases in Fig.~\ref{fig:VLFT_FiniteBlock}, expected throughput for the finite-$N$ (repeated) VLFT codes converges to that of VLFT with $N = \infty$ before expected latency has reached $200$ symbols. The penalty of $\Delta=0.3C$ compared to $\Delta=0.4C$ is only visible when $M$ is small. 

VLFT codes can have expected throughput higher than the original BSC capacity because of the beneficial effect of the error-free termination symbol.  This effect becomes smaller as expected latency increases.
	 \begin{figure}[t]
	 \centering
    \includegraphics[width=0.5\textwidth]{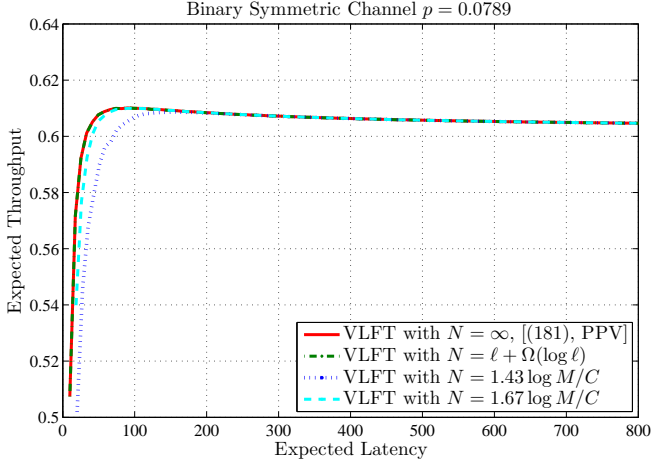}
    \caption{Performance comparison of VLFT code achievability based on RCU bound with different codebook block-lengths.}
    \label{fig:VLFT_FiniteBlock}
    \end{figure}~
	 \begin{figure}[t]
	 \centering
    \includegraphics[width=0.5\textwidth]{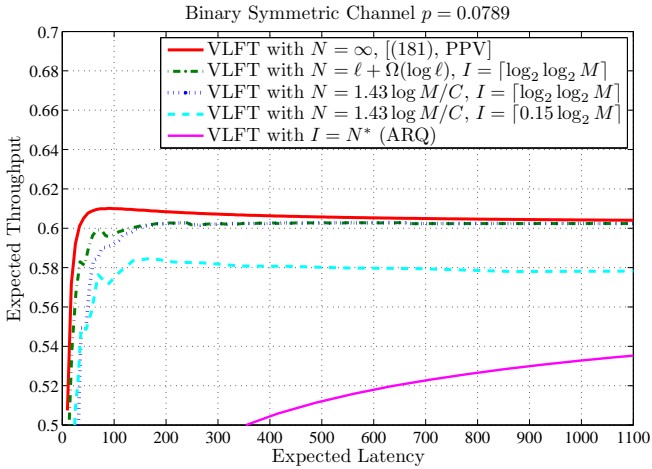}
    \caption{Performance comparison of VLFT code achievability based on RCU bound with uniform increment and finite-length limitations.}
    \label{fig:VLFT_UnifInc}
    \end{figure}

Fig.~\ref{fig:VLFT_UnifInc} shows the performance of the repeated VLFT code with various decoding-time increments $I$.  As in \eqref{eq:Thm2}, when $I$ grows linearly with $\log M$ (i.e. $\lceil 0.15\log_2 M\rceil$)  then there is a constant gap from the $I=1$ case.  However, if $I$ grows as $\lceil\log_2 \log_2 M\rceil$ then the gap from the $I=1$ case decreases as expected latency increases.  ARQ performance (in which $I=N*$, an optimized block-length) is also shown in the figure, which reveals a considerable performance gap from even the most constrained VLFT implementation we implemented.


\section{Conclusion}
\label{sec:Conclusion}
This paper shows that the achievable performance of a VLFT code is mostly preserved when the block-length of the underlying code is restricted to be finite and decoding attempts are limited to regularly spaced decoding times.  Specifically, if block-length $N = \ell + \Omega(\log\ell)$ and $I = O(\log\ell)$ the optimal expansion of $M_t^*$ is achieved.  

The finite-block-length results for VLFT codes suggest that it is not necessary to use an infinitely long codebook or even a very large one. Numerical results show that a base code rate that is $67\%$ of the capacity can closely approach performance of a VLFT code with an infinite block-length.  Numerical results also show that decoding after every $\log_2\log_2 M$ symbols is almost as good as decoding at every symbol. 
 
\setcounter{equation}{85}
\begin{figure*}[!t]
\normalsize
\begin{align}
\label{eqn:6.1}
&(1-\P[\zeta_{N}]) \ell
\leq n_1 + I\sum_{j = 1}^{m-1}\E[\exp\{-[S_{n_j} - \log M]^+\}]
\\
&=  \E[n_1 + (m-1)I;\tau_0 \geq m] + \E[n_1+(\tau_0-1)I;\tau_0 < m]
+ I\E\left[\sum_{k = 0}^{m-1-\tau_0}\exp\left\{-\left[S_{n_{k+\tau_0}} - \log M\right]^+\right\};\tau_0 < m\right]
\\
\label{eqn:6.3}
& \leq \left[n_1 + (m-1)I\right] \P[\tau_0 \geq m] + I\E[(\tau_0-1) ;\tau_0 < m]
+ I\E\left[\sum_{k = 0}^{m-1-\tau_0}\exp\left\{-\left[S_{n_{k+\tau_0}} - \log M\right]^+\right\};\tau_0 < m\right]
\\
& \leq N\P[\tau_0 \geq m] + \E[n_{\tau_0}; \tau_0 < m]
+ I\E\left[\sum_{k = 0}^{m-1}\exp\left\{-\left[S_{n_{k}}\right]^+\right\};\tau_0 < m\right]
\\
\label{eqn:6.4}
&\leq \frac{\log M}{C} + N\P[\tau_0 \geq m] + O(I).
\end{align}
\hrulefill
\end{figure*}
\setcounter{equation}{69}
\section{Appendix}
\label{sec:Addendix}
\begin{proof}[Proof of Thm. \ref{thm:FiniteVLFT}]
Consider a random codebook $\mc{C}_N = \{\ms{C}_1, \dots, \ms{C}_M\}$ with $M$ codewords of length-$N$ and codeword symbols independent and identically distributed according to $P_X$. To construct a VLFT code consider the following $(U, f_n, g_n, \tau)$:
The common random variable
\begin{align}
U \in \mc{U} = \overbrace{\mc{X}^N\times\dots\times\mc{X}^N}^{M \text{times}}.
\end{align}
is distributed as:
\begin{align}
U \sim \prod_{j=1}^M P_{X^N}.
\end{align}
A realization of $U$ corresponds to a deterministic codebook  $\{\ms{c}_1, \dots, \ms{c}_M\}$. 
Let $x(n)$ denote the $n$th coordinate of a vector $x$. The sequence $(f_n, g_n)$ is defined as
\begin{align}
f_n(U, W) &= \ms{C}_W(n)
\\
g_n(U, W, Y^n) &=  \arg\max_{j = 1, \dots, M} i(\ms{C}_W(n); Y^n)
\end{align}
and the stopping time $\tau$ is defined as:
\begin{align}
\tau &= \inf\{n: g_n(U,Y^n) = W\}\wedge N\,.
\end{align}
The $n$th marginal error event $\zeta_n$ is given as:
\begin{align}
\zeta_n = \left\{\bigcup_{j \ne W}i(\ms{C}_j^n; Y^n) > i(\ms{C}_w^n; Y^n)\right\}\,.
\end{align}
Following \eqref{eqn:SumsOfPtau}-\eqref{eqn:SumsOfJoints} we have 
\begin{equation}
\E[\tau] = \sum_{n = 0}^{N-1}  \P[\tau > n] \leq \sum_{n = 0}^{N-1}  \P[\zeta_n] \leq \sum_{n = 0}^{N-1} \xi_n\,.
\end{equation}
As in \cite[(151)-(153)]{PolyIT11}, the last inequality follows from union bound and the fact that a probability measure is upper bounded by $1$.
With a similar bounding technique, the error probability can be upper bounded as:
\begin{align}
\P[g_\tau(U, Y^\tau)\ne W]
&= \P[g_N(U, Y^N) \ne W, \tau = N]
\\
&= \P\left[\bigcap_{j = 1}^N \zeta_j\right]
\\
&\leq \P[\zeta_N]
\\
&\leq  \xi_N\,.
\end{align}
In other words, the error probability is upper bounded by the error probability of the base code $\mc{C}_N$.
\end{proof}

\begin{proof}[Proof of Thm. \ref{thm:FiniteFV}]
The proof follows from random coding and the following modification of the triplet $(f_n, g_n, \tau)$ of Thm. \ref{thm:FiniteVLFT}: For $k = 1, 2, \dots $ let $(f'_n, g'_n)$ be defined as:
\begin{align}
\nonumber
f'_n(U,W) &= \begin{cases} 
   f_n(U, W) &\text{ if } n \leq N  \\
   f_{n-kN}(U, W)    &\text{ if } kN < n \leq (k+1)N
 \end{cases}
\\
\nonumber
g'_n(U,Y^n) &= \begin{cases} 
   g_n(U,Y^n)  &\text{ if } n \leq N  \\
   g_{n-kN}(U,Y_{kN+1}^n)   &\text{ if }  kN < n \leq (k+1)N
\end{cases}
\end{align}
Let the new stopping $\tau'$ be defined as:
\begin{align}
\tau' &= \inf\{n: g'_n(U,Y^n) = W\}\,.
\end{align}
The zero-error part is obvious from the definition of the stopping time $\tau'$. As mentioned above, the new encoder/decoder sequence $(f'_n, g'_n)$ is simply an extension of the VLFT code in Thm.~\ref{thm:FiniteVLFT} by performing an ARQ-like repetition. The expectation of $\tau'$ is thus given as:
\begin{align}
\E[\tau'] &=\sum_{n = 1}^{N-1} \P\left[\bigcap_{j = 1}^{n}\zeta_j\right] + \P\left[\bigcap_{j = 1}^{N}\zeta_j\right]\E[\tau']
\\
&\leq \sum_{n = 1}^{N-1} \P[\zeta_n] + \P[\zeta_N]\E[\tau'] \,,
\end{align}
which implies that:
\begin{equation}
 \E[\tau'] \leq (1-\P[\zeta_N])^{-1}\sum_{n = 1}^{N-1} \P[\zeta_n] \,.
\end{equation}
Applying RCU bound on $\P[\zeta_n]$ for each $n$ finishes the proof.  
\end{proof}
\begin{proof}[proof of Thm. \ref{thm:MainAsympResult1}]
Consider the FV code as in Thm.~\ref{thm:FiniteFV} 
but with an initial-block-length $n_1$, an uniform increment $I$ and a finite $m$. The finite block-length is given by $N = n_m$ where $n_j = n_1 + (j-1)I$. Define the auxiliary stopping time as:
\begin{align}
\tau_0 = \inf\{j > 0: i(X^{n_j}; Y^{n_j}) \geq \log M\} \wedge m \,.
\end{align}
Similar to Thm. \ref{thm:FiniteFV} we have \eqref{eqn:6.1} to \eqref{eqn:6.4}, shown at the top of the page. 
Now we are left to choose the scaling of $m$. Using a similar choice as in the proof of Thm. \ref{thm:VLFTExpandFinite}:
\setcounter{equation}{90}
\begin{align}
m = \left\lceil{\left(\frac{\log M}{IC_\Delta} - 
\frac{n_1}{I}\right) + 1}\right\rceil
\end{align}
which yields 
\begin{align}
\label{eqn:deltaforM}
N = n_1 + (m-1)I \geq \frac{\log M }{C_\Delta}\,.
\end{align}
Rest of the proof follows as in the proof of Thm. \ref{thm:VLFTExpandFinite}.
\end{proof}


\bibliographystyle{IEEEtran}
\bibliography{IEEEabrv,Feedback_Journal}

\end{document}